\documentclass[11pt]{amsart}
\usepackage{geometry}                
\usepackage{graphicx}
\usepackage{amssymb, amsthm, mathrsfs}
\usepackage{epstopdf}
\usepackage{enumerate}
\usepackage{color}
\usepackage{hyperref}
\newcommand{\be}{\begin{equation*}}
\newcommand{\ee}{\end{equation*}}
\newcommand{\ben}[1]{\begin{equation}\label{#1}}
\newcommand{\een}{\end{equation}}
\newcommand{\bea}{\begin{eqnarray}}
\newcommand{\eea}{\end{eqnarray}}

\newcommand{\half}{{\frac{1}{2}}}
\newcommand{\R}{\mathbb{R}}

\renewcommand{\u}{\mathbf{u}}
\newcommand{\f}{\mathbf{f}}
\renewcommand{\v}{\mathbf{v}}
\renewcommand{\O}[1]{\mathcal{O}\left( #1 \right)}

\renewcommand{\H}{\underline{H}}
\newcommand{\Ht}{\tilde{\underline{H}}}

\renewcommand{\L}{\underline{L}}
\newcommand{\wto}{\rightharpoonup}

\newcommand{\da}{\partial^{\alpha}}

\newcommand{\dad}{\partial^{1-\alpha}}
\newcommand{\scri}{\mathscr{I}}
\newcommand{\esssup}{\mathop{\textrm{ess\,sup}}}

\newcommand{\cc}{\subset \subset}
\newcommand{\abs}[1]{\left|#1 \right|} 
\newcommand{\norm}[2]{\left|\left |#1 \right| \right |_{#2}} 
\newcommand{\ip}[3]{\left(#1,#2 \right )_{#3}} 
\newcommand{\pair}[2]{\left\langle#1,#2 \right \rangle} 
\newcommand{\Ub}{\overline{U}}
\newcommand{\dU}{\partial{U}}

\newcommand{\eq}[1]{(\ref{#1})}
\newtheorem{Theorem}{Theorem}

\newtheorem{Lemma}{Lemma}
\newtheorem{Corollary}[Theorem]{Corollary}
\newtheorem{Definition}{Definition}
\numberwithin{equation}{section}
\numberwithin{Theorem}{section}
\numberwithin{Lemma}{subsection}

\title{The massive wave equation in asymptotically AdS spacetimes}
\author{C. M. Warnick}
\thanks{\texttt{warnick@ualberta.ca} \\ \phantom{1 } Department of Physics, 4-183 CCIS, University of Alberta, Edmonton AB T6G 2E, Canada}

\begin{document}

\begin{abstract}
We consider the massive wave equation on asymptotically AdS spaces. We show that the timelike $\scri$ behaves like a finite timelike boundary, on which one may impose the equivalent of Dirichlet, Neumann or Robin conditions for a range of (negative) mass parameter which includes the conformally coupled case. We demonstrate well posedness for the associated initial-boundary value problems at the $H^1$ level of regularity. We also prove that higher regularity may be obtained, together with an asymptotic expansion for the field near $\scri$. The proofs rely on energy methods, tailored to the modified energy introduced by Breitenlohner and Freedman. We do not assume the spacetime is stationary, nor that the wave equation separates.\\\phantom{1}\\ \phantom{1} \hfill ALBERTA THY 3-12
\end{abstract}

\maketitle

\section{Introduction}

Among the solutions of Einstein's general theory of relativity, the maximally symmetric spacetimes  hold a privileged position. Owing to the high level of symmetry, they serve as plausible `ground states' for the gravitational field, so there is great interest in spacetimes which approach a maximally symmetric spacetime in some asymptotic region. Such a spacetime would represent an `isolated gravitating system'. Historically, asymptotically flat spacetimes have been the most studied, however, recently there has been great interest in the asymptotically anti-de Sitter (AdS) spacetimes motivated by the putative AdS/CFT correspondence \cite{Maldacena}. In the study of classical General Relativity, there have also been some very interesting recent results regarding the question of black hole stability \cite{Holzegel:2011rk, Holzegel:2011uu, Holzegel:2011qk} for asymptotically AdS black holes.

The asymptotically AdS spacetimes approach (the covering space of) the spacetime of constant  sectional curvature $-\frac{3}{l^2}$, which we shall refer to as the AdS spacetime. In so called `global coordinates', the metric takes the form
\be
g = -\left(1+\frac{r^2}{l^2} \right) dt^2 + \frac{dr^2}{1+\frac{r^2}{l^2}} + r^2 (d\theta^2 + \sin^2\theta d\phi^2).
\ee
In contrast to the Minkowski spacetime, AdS has a timelike conformal boundary, $\scri$. Accordingly one expects that in order to have a well posed time evolution for the equations of physics in this background it is necessary to specify some boundary condition on $\scri$. In pioneering work \cite{BF}, Breitenlohner and Freedman considered the massive wave equation
\ben{eqn1}
\Box_g \phi - \frac{\lambda}{l^2} \phi = 0,
\een
on a fixed anti-de Sitter background of constant sectional curvature $-\frac{3}{l^2}$. They were able to solve the wave equation by separation of variables, making use of the $SO(2,3)$ symmetry of AdS. The second order ordinary differential equation governing the radial part of the wave equation has a regular singular point at infinity, hence the field has an expansion near infinity:
\be
\phi(r,\theta ,\phi, t) = \frac{1}{r^{\lambda_+}}\left[ \psi^+(\theta ,\phi, t)+ \O{\frac{1}{r^2}} \right]+ \frac{1}{r^{\lambda_-}}\left[ \psi^-(\theta ,\phi, t)+  \O{\frac{1}{r^2}}\right]
\ee
where $\lambda_\pm = \frac{3}{2} \pm \sqrt{\frac{9}{4}+\lambda}$.  When the mass parameter is in the range $-\frac{9}{4}<\lambda <0$,
\emph{both} branches decay towards infinity. For a well posed problem it is necessary to place some constraints on the functions $\psi^\pm$. The usual choice would be to insist that $\psi^-=0$, which is analogous to imposing a Dirichlet condition at $\scri$. This corresponds to requiring a solution of finite energy (we shall elaborate on this point later). Breitenlohner and Freedman showed that for $-\frac{9}{4}<\lambda <-\frac{5}{4}$ the wave equation can also be solved on the exact anti-de Sitter space under the assumption that $\psi^+=0$, analogous to a Neumann condition.\footnote{Similar considerations hold in higher dimensions, where this bound becomes $-\frac{n^2}{4}<\lambda< -\frac{n^2}{4}+1$, where $n+1$ is the spacetime dimension. From now on, we will assume $\lambda$ to lie in this range.} Breitenlohner and Freedman also introduced a modified or renormalised energy which is finite for both branches of the solution. As in the case of a finite domain, the Dirichlet and Neumann boundary conditions are not the only possible choices. We summarise some other possible boundary conditions in Table \ref{bctab} below (the list is not exhaustive).

\begin{table}[htdp] \label{bctab}
\begin{center}\begin{tabular}{|c|l|c|}
\hline (1) & Dirichlet & $\psi^-=0$  \\\hline (2) &  Inhomogeneous Dirichlet & $\psi^- = f$ \\\hline (3) &Neumann &$\psi^+ = 0$\\\hline (4)  &Inhomogeneous Neumann &$\psi^+ = f$  \\\hline (5)  &Robin &$\psi^++\beta \psi^- = 0$  \\\hline (6)  &Inhomogeneous Robin &$\psi^++\beta \psi^- = f$  \\\hline \end{tabular} 
\end{center}
\caption{Possible boundary conditions. $f$, $\beta$ are functions on $\scri$.}
\end{table}

 The homogeneous conditions (1), (3), and  (5) were considered by Ishibashi and Wald \cite{Ishibashi:2003jd, Ishibashi:2004wx}, who showed that they give rise to a well defined unitary evolution for the scalar wave, Maxwell and gravitational perturbation equations in the exact anti-de Sitter spacetime. This work has been extended to the Dirac equation in the work of Bachelot \cite{Bachelot}. These papers use methods based on self-adjoint extensions of the elliptic part of the wave operator. They make crucial use of properties of the exact AdS space (in particular staticity and separability) which are not shared by the general class of asymptotically AdS spaces. For Dirichlet boundary conditions (1), (2), the work of Holzegel \cite{Holzegel:2011qj} (using energy space methods) and Vasy \cite{Vasy} (using microlocal analysis) provides well posedness results for a more general class of asymptotically AdS spaces for the range $\lambda>-\frac{9}{4}$, but does not treat the other possible boundary conditions. 
 
 The aim of this work is to treat all of the boundary conditions (1)-(6), without making any assumptions regarding the staticity or separability of the metric. The natural approach is that of energy estimates, however, we are forced to confront the problem that for boundary conditions (2)-(6) we cannot expect the standard energy to be finite. To deal with this, we use the renormalised energy of \cite{BF}. We show that the natural Hilbert spaces associated with this energy are generalisations of the usual Sobolev spaces where `twisted' derivatives of the form
 \be
 \da_i f:= \rho^{-\alpha}  \frac{\partial}{\partial x_i}\left( \rho^\alpha f\right)
 \ee
 are supposed to exist in a distributional sense and belong to an appropriate $L^2$ space, for some appropriately chosen $\rho$, related to the distance to the boundary.
 
 To simplify the analysis at the expense of losing the geometrical structure, we map the problem to a more general problem in a finite region of $\mathbb{R}^N$ which is very closely analogous to that of a finite initial boundary value problem (IBVP). We introduce weak formulations at the $H^1$ level for all of the boundary conditions, and are able to show that these weak problems admit a unique solution. The method used to show existence is to approximate the problem by a suitable hyperbolic IBVP in a finite cylinder, and let the cylinder approach $\scri$. We then use energy estimates to extract a weakly convergent subsequence. in this way, we can show well posedness results for (1)-(6). We are also able to recover higher regularity for the solution, if more assumptions are made on the data. We further provide an asymptotic expansion for the solutions near infinity.

The paper will be structured as follows. We first define the asymptotically AdS spaces we consider in \S \ref{AsyAdS}. We then introduce the modified energy in \S \ref{Energy} and use it in \S \ref{WP} to motivate weak formulations of the Dirichlet and Neumann problems, which we then show to be well posed. In \S \ref{HR} we show that under stronger assumptions on data improved regularity can be obtained, together with the asymptotic behaviour of the solution. Finally, in \S \ref{bcs} we discuss briefly the inhomogeneous and Robin boundary conditions and remark on the connection to methods involving self-adjoint extensions. We assume throughout a degree of familiarity with the theory of the finite IBVP, as developed for example in \cite{Evans, Lady}.

\subsection*{Acknowledgements}
I would like to thank Gustav Holzegel for introducing me to this problem, and for  helpful comments. I would also like to thank Mihalis Dafermos, Julian Sonner, Pau Figueras as well as the anonymous referees for comments. I would like to acknowledge funding from PIMS and NSERC. The early stages of this project were supported by Queens' College, Cambridge.

\section{Asymptotically AdS spaces\label{AsyAdS}}

\begin{Definition} \label{def1}
Let $X$ be a $n+1$ dimensional manifold with boundary\footnote{We take the convention that $X$ includes $\partial X$ as a point set, while $\mathring{X} = X \setminus \partial X$.} $\partial X$, and $g$ be a smooth Lorentzian metric on $\mathring{X}$. We say that a connected component $\scri$ of $\partial X$ is an \emph{asymptotically anti-de Sitter end of $(\mathring{X}, g)$ with radius $l$} if:
\begin{enumerate}[i)]
\item There exists a smooth function $r$ such that $r^{-1}$ is a boundary defining function for $\scri$.
\item If $x^\alpha$ are coordinates on the slices $r=\textrm{const.}$, we have locally
\be
g_{rr} = \frac{l^2}{r^2}+ \O{\frac{1}{r^4}},\qquad g_{r\alpha} = \O{\frac{1}{r^2}}, \qquad g_{\alpha\beta} = r^2 \mathfrak{g}_{\alpha \beta} + \O{1},
\ee
where $\mathfrak{g}_{\alpha\beta}dx^\alpha dx^\beta$ is a Lorentzian metric on $\scri$.
\item $r^{-2} g$ extends as a smooth metric on a neighbourhood of $\scri$.
 \end{enumerate}
We say that $r$ is the asymptotic radial coordinate and $\scri$ is the conformal infinity of this end.
\end{Definition}

We make here a few remarks about these assumptions
\begin{enumerate}[1.]
\item Note that $r$ and $\mathfrak{g}$ are not unique. A different choice of $r$ gives rise to a different $\mathfrak{g}$ conformally related to the first, hence the nomenclature `conformal infinity'.

\item Condition $ii)$ can be weakened to $g_{rr} = l^2 r^{-2} + \O{r^{-3}}, g_{\alpha\beta} = r^2 \mathfrak{g}_{\alpha\beta} + \O{r}, g_{r\alpha} = \O{\frac{1}{r}}$ for well posedness of the massive wave equation, however one then needs to make a more careful choice of twisting function $\rho$. For simplicity of exposition, we do not consider this possibility.

\item \label{com3} Condition $iii)$, sometimes known as weak asymptotic simplicity, is not necessary for the weak well posedness of the massive wave equation\footnote{$C^2$ extensibility certainly suffices, $C^{1, \gamma}$ is probably enough}, but is necessary to obtain the full asymptotic expansion for the scalar field near $\scri$. In particular this condition implies that taking radial derivatives of the metric functions improves radial fall-off by $r^{-1}$, while taking tangential derivatives does not change the asymptotics.

\end{enumerate}

\section{The modified energy\label{Energy}}

We will now consider for a moment the case of the exact AdS spacetime. The usual energy one associates to solutions of the massive wave equation is given by
\ben{eng1}
E[\Sigma_t] = \int_{\Sigma_t} T_{\mu \nu} K^\mu dS^\nu,
\een
where $K^\mu$ is the timelike Killing vector and the energy-momentum tensor is given by
\ben{eng2}
T_{\mu \nu} = \nabla_{\mu} \phi \nabla_\nu \phi - \frac{1}{2}g_{\mu \nu}\left( \nabla_\sigma \phi \nabla^\sigma \phi + \frac{\lambda}{l^2} \phi^2 \right),
\een
and satisfies $\nabla_\mu T^{\mu \nu} = 0$ when $\phi$ is a solution of \eq{eqn1}. If $\phi$ has Dirichlet decay then one expects, by power counting, that $E[\Sigma_t]$ will be finite. However, if we have Neumann decay then the integral in \eq{eng1} fails to converge near infinity. In order to deal with this problem, Breitenlohner and Freedman modified the energy momentum tensor \eq{eng2} to give a new tensor
\ben{eng3}
\tilde{T}_{\mu \nu} = T_{\mu \nu} + \Delta T_{\mu \nu}= T_{\mu \nu} +\kappa (g_{\mu \nu} \Box - \nabla_\mu \nabla_{\nu} + R_{\mu \nu}) \phi^2.
\een
The modification satisfies
\be
\nabla_\mu( \Delta T^{\mu}{}_\nu) = \frac{\kappa}{2}(\partial_\nu R) \phi^2
\ee
and so for the exact AdS spacetime, $\tilde{T}_{\mu \nu}$ will also give rise to a formally conserved energy. It transpires that the new energy $\tilde{E}[\Sigma_t]$ differs from the original energy $E[\Sigma_t]$ by a surface term. This surface term vanishes when $\phi$ decays like the Dirichlet branch, but diverges when $\phi$ decays like the Neumann branch. By choosing $\kappa$ appropriately, it is possible to construct an energy which is positive, finite and conserved for both Dirichlet and Neumann decay conditions.

Our proofs will make use of the modified energy associated to the timelike vector $\partial_t$ which, however, will only be approximately conserved since we no longer assume an exactly stationary spacetime. Rather than using the definition \eq{eng3}, it is in practice more straightforward to work directly from the PDE. In order to see how this works, we will briefly discuss a toy model which captures the salient features of the problem.

\subsection{A toy model\label{toy}}
 
Consider the wave equation
\ben{eq1}
u_{tt}-u_{xx}-\frac{u_x}{x}+\alpha^2 \frac{u}{x^2}=0, \qquad 0<x \leq 1,
\een
where $0< \alpha <1$, subject to initial conditions
\bea
u(x, 0) &=& u_0(x) \nonumber\\
u_t(x,0) &=& u_1(x) \label{eq2}.
\eea
Considering the behaviour near $x=0$, we hope to impose as boundary conditions either
 \be
 \begin{array}{rcl}
 u &\sim& x^{\alpha}[1+\O{x^2}] \quad \textrm{(Dirichlet)}\\ u &\sim& x^{-\alpha}[1+\O{x^2}] \quad \textrm{(Neumann)}.
 \end{array}
 \ee
 At $x=1$ we will require that $u=0$.
 
Suppose we have a suitably smooth solution to this equation. We can multiply \eq{eq1} by $x u_t$ and, after integrating by parts, deduce the conservation law for the standard energy:
\ben{eq4}
\frac{dE}{dt} = \frac{d}{dt} \half \int_0^1 \left(u_t^2 +u_x^2 +  \alpha^2 \frac{u^2}{x^2}\right)x dx = \left [ x u_x u_t\right]_0^1.
\een
For Dirichlet boundary conditions at $x=0$ the right hand side vanishes, however, for Neumann boundary conditions it is infinite. In order to introduce the modified energy, it is convenient to re-write the equation in the following form
\ben{eq6}
u_{tt} - x^{-1+\alpha}\frac{\partial }{\partial x}\left( x^{1-2 \alpha} \frac{\partial}{\partial x}(x^\alpha u)\right)=0,
\een
which gives \eq{eq1} upon expanding using Leibniz rule. We can multiply \eq{eq6} by $xu_t$ and integrate by parts to deduce
\ben{eq10}
\frac{d\tilde{E}}{dt} = \frac{d}{dt} \half \int_0^1 \left(u_t^2 +[x^{-\alpha}\partial_x(x^\alpha u)]^2\right)x dx = \left [ x u_t x^{-\alpha}\partial_x(x^\alpha u) \right]_0^1.
\een
Now, for Dirichlet conditions at $x=0$ we again find that the right hand side vanishes, however we now find that it also vanishes for the Neumann behaviour at $x=0$. The two energies differ by a surface term:
\be
\tilde{E}-E = \frac{1}{2}\left[ \alpha u^2 \right]_0^1,
\ee
which vanishes for Dirichlet conditions at $x=0$ and is infinite for the Neumann conditions. In this sense we can view $\tilde{E}$ as a `renormalized' energy, since we have formally subtracted an infinite boundary term from the infinite energy to get a finite result.

Thus even though the standard energy is infinite for the Neumann behaviour, we can modify it to get a conserved, finite, positive energy. We see now the justification for using the terms `Dirichlet' and `Neumann' to describe the boundary conditions. The Dirichlet condition requires $u \to 0$ as $x \to 0$, while the Neumann condition requires $x^{-\alpha}\partial_x(x^\alpha u)\to 0$ as $x \to 0$.

This discussion also suggests that it will be fruitful to re-formulate the equation in terms of `twisted' derivatives of the form $x^{-\alpha}\partial_x(x^\alpha \cdot)$. We shall do so in the next section and this will lead us to the appropriate setting in which to discuss the well posedness of \eq{eqn1}.

\section{Well Posedness of the Weak Formulation\label{WP}}

\subsection{Defining the problem}
Motivated by the discussion of the previous section, we can now define the framework in which we shall work.

We assume that $U\subset\mathbb{R}^N$ is a bounded subset of $\mathbb{R}^N$ with compact $C^\infty$ boundary $\partial U$. This means that in the neighbourhood of any point $P \in \partial U$, there exists an open neighbourhood $W_P \subset \Ub$ of $P$ and a smooth bijection $\Phi_P: W_P \to \mathbb{R}^N_+ \cap B(\mathbf{0}, \delta_P)$, where $\mathbb{R}^N_+ = \{ (x, x^a) \in \mathbb{R}^N : x\geq 0\}$ and $B(\mathbf{x},r)$ is the open Euclidean ball centred at $\mathbf{x}$ with radius $r$. 

We're also going to assume that there exists a smooth function $\rho:\Ub \to \mathbb{R}_+$, which vanishes only on $\partial U$ and such that there exists $\tilde\epsilon$ so that if $d(x, \partial U)<\tilde\epsilon$, we have $\rho(x) = d(x,\partial U)$ and if $d(x, \partial U)>\tilde\epsilon$, $\rho(x)>\tilde\epsilon$. We will set $n_i = \partial_i \rho$, which extends the unit normal of $\dU$ into the interior of $U$. We may assume that the neighbourhoods $W_P$ are such that $\rho \circ \Phi_P^{-1}(x, x^a) = x$ and $\delta_P = \tilde \epsilon$. We denote by $U_T$ the timelike cylinder $(0, T) \times U$ and by $\dU_T$ the boundary $(0, T) \times \dU$.

We define our twisted derivatives in a similar vein to above. For a differentiable function, we set
\be
\da_i u =\rho^{-\alpha}\frac{\partial}{\partial x^i}\left(\rho^\alpha u \right).
\ee
Throughout, we will assume $0<\alpha<1$. We can see that this restriction is necessary from the toy model, since if $\alpha$ is outside this range, only the Dirichlet behaviour is compatible with finite energy even after renormalisation.

We may now define the equation in which we are interested:
\ben{WP1}
 u_{tt} +\mathcal{L} u= f\textrm{ in } U_T ,
 \een
where\footnote{We use the Einstein summation convention, so that repeated indices $i, j$ etc. should be summed over.}
\be
\mathcal{L} u =  -\dad_i \left( a_{ij} \da_j u\right ) + b_i \da u + c u 
\ee
subject to the initial conditions
\ben{WP2}
u(x, 0) = u_0, \qquad u_t(x, 0) = u_1.
\een
We assume all coefficients $a_{ij}, b_i, c$ are in $C^\infty(\Ub_T)$, however this is certainly stronger than necessary\footnote{see Comment \ref{com3} after Definition \ref{def1}}. We will assume throughout that $a_{ij}$ is a symmetric matrix such that the uniform ellipticity condition holds:
\ben{WP5}
\theta \abs{\xi}^2 \leq a_{ij} \xi^i \xi^j
\een
for any $\xi^i\in \mathbb{R}^N$, where $\theta$ is uniform in both time and space coordinates, and furthermore that $n_ia_{ij}$ is independent of $t$, \emph{on the boundary $\dU$}.

For the time being, there are two possible boundary conditions in which we shall be interested. We will consider both Dirichlet:
\ben{WP3}
u = 0\ \ \textrm{ on }\ \  \partial U_T,
\een
and Neumann:
\ben{WP4}
 n_i a_{ij} \da_{j} u  = 0\ \ \textrm{ on }\ \  \partial U_T,
\een
boundary conditions. 

To justify considering this equation, we have the following Lemma
\begin{Lemma}
Suppose $\scri$ is an asymptotically AdS end of $(\mathring{X}^{n+1}, g)$ with radius $l$, and let $\mathcal{P}\in\scri$. Then there exists a smooth Lorentzian metric , $\tilde{g}$, on the solid cylinder $U_T = [-T, T] \times B(0, 1)\subset \R^{n+1}$ together with a neighbourhood of $\mathcal{P}$ which embeds isometrically into $(U_T, \tilde{g})$, with $\scri$ mapped to (a portion of) the boundary of the cylinder. Furthermore, setting $\phi = \frac{p}{r^{n/2}}  u$ for some $p\in C^\infty(\Ub_T)$ depending only on $g$, the wave equation
\ben{WP6}
\Box_{\tilde{g}} \phi - \frac{\lambda}{l^2} \phi = 0
\een
may be cast in the form \eq{WP1} for some $\rho$, $a_{ij}$, $b_i$, $c$ satisfying the assumptions above, with $\rho = r^{-1} + \O{r^{-3}}$ and\footnote{
Note that the bound $-\frac{n^2}{4}<\lambda< -\frac{n^2}{4}+1$ implies $0<\alpha<1$.} $\alpha = \sqrt{\frac{n^2}{4}+\lambda}$.
\end{Lemma}
\begin{proof}
Define $s = r^{-1}$, so that $\scri = \{s=0\}$ and $\hat{g}=s^2 g$ is a smooth metric on $X$, with $\scri$ a totally geodesic submanifold. Pick a spacelike surface $\Sigma_0$, containing $\mathcal{P}$ such that, $\Sigma_0$ is orthogonal to $\scri$ and has normal $n_{\Sigma_0}$ with respect to $\hat{g}$. We can push forward $n_{\Sigma_0}$ using the geodesic flow of $\hat{g}$ on $TX$ to give a smooth unit vector field $T$, with associated diffeomorphism $\psi_t$, in a neighbourhood of $\mathcal{P}$. Now pick coordinates $x^i = (\rho, x^a)$ on $\Sigma_0$ near $\mathcal{P}$, such that $\rho =s|_{\Sigma_0}$, and extend them off $\Sigma_0$ by requiring $T\rho=Tx^a=0$. Near $\mathcal{P}$ we may take as coordinate functions $(t, \rho, x^a)$, and in these coordinates, the metric coefficients satisfy
\bea
& g_{tt} = -\frac{1}{\rho^2}, \qquad g_{t a} = 0, &\nonumber \\
& g_{\rho\rho} = \frac{l^2}{\rho^2} + \O{1}, \qquad g_{\rho a} =  \O{1}, \qquad g_{ab} =  \frac{h_{ab}}{\rho^2} + \O{1},&
\eea
as $\rho \to 0$, for some $h_{ab}$ independent of $\rho$. We can assume that the coordinate neighbourhood, $V$ we constructed is in fact contained in a coordinate neighbourhood of  the boundary of $U_T$, and extend the metric $\hat{g}$ smoothly to a metric on the whole of $U_T$, and define $\tilde{g} = \rho^{-2} \hat{g}$. This agrees with $g$ in $V$.

Now note that $\sqrt{g} = \rho^{-(n+1)} (\sqrt{h}+ \O{\rho^2})$. We set
\be
\phi = g^{-1/4} \rho^{-n-\frac{1}{2}} u,
\ee
and it may then be verified that \eq{WP6} may be cast in the form \eq{WP1}, with $\alpha = \sqrt{\frac{n^2}{4}+\lambda}$,  $a_{ij} = \rho^{-2} g^{ij}$ and $b_i$, $c$ similarly given by functions constructed from the metric and its derivatives which are smooth up to $\rho=0$. 
\end{proof}

Making use of the finite speed of propagation for solutions of hyperbolic equations, any well posedness results for the problem \eq{WP1} may be extended to regions of $\mathring{X}$, assuming some global causality conditions. In particular, well posedness of \eq{WP1} with appropriate boundary conditions implies well posedness in the region $D^+[\Sigma \cup (I^+(\Sigma) \cap \scri)]$ for any spacelike hypersurface $\Sigma$, with initial data specified on $\Sigma$.

\subsection{The function spaces}

In order to introduce a weak formulation for the initial-boundary value problem we are considering, it will be necessary to define the function spaces in which we seek a solution. For a locally measurable function $u$ and measurable set $V \subset U$, we define the norm and space
\ben{WP7}
\norm{u}{\L^2(V)}^2 = \int_V u^2 \rho dv. \qquad \L^2(V) = \{u: \norm{u}{\L^2(V)}<\infty\},
\een
where $dv$ is the Lebesgue measure. This is clearly a Hilbert space with inner product
\ben{WP8}
\ip{u_1}{u_2}{\L^2(V)} = \int_V u_1 u_2 \rho dv.
\een
Now, we note that for smooth functions $\phi, \psi$ of compact support we may integrate by parts to find
\be
\int_V \phi \da_i \psi \rho dx = -\int_V (\dad_i \phi) \psi \rho dx, \qquad i=1, \ldots, N
\ee
This allows us to define a weak version of $\da$. We say that $v_i=\da_i u$ is the weak $\alpha$-twisted derivative of $u$ if
\ben{WP9}
\int_V v_i \phi \rho dx = - \int_V u \dad_i \phi \rho dx
\een
for all $\phi \in C_c^\infty(V)$. We say that $u \in \H^1(V)$ if $\da u$ exists in a weak sense and $\da_i u \in \L^2(V)$. We can define a norm and inner product on $\H^1(V)$ as follows:
\bea
\norm{u}{\H^1(V)}^2 &=& \norm{u}{\L^2(V)}^2+ \sum_{i=1}^N\norm{\da_i u}{\L^2(V)}^2, \nonumber \\
\ip{u_1}{u_2}{\H^1(V)} &=& \ip{u_1}{u_2}{\L^2(V)} +\sum_{i=1}^N \ip{\da_i u_1}{\da_i u_2}{\L^2(V)}. \label{WP10}
\eea
Next we define $\H_0^1(V)$ to be the completion of $C_c^\infty(V)$ with respect to the norm $\norm{\cdot}{\H^1(V)}$. We shall often take $V=U$. On any subset compactly contained in $U$, these spaces are simply equivalent to the standard Sobolev spaces.

We note at this stage that $\dad_i$ is the formal adjoint of $\da_i$ with respect to the $\L^2$ inner product. Thus the second order operator $\dad_i\left( a_{ij} \da_j \cdot\right)$ appearing in \eq{WP1} is formally self-adjoint. When we come to consider higher regularity, we shall need the Sobolev space associated to the adjoint derivative operator. In particular $u \in \Ht^1(V)$ if $\dad_i u$ exist in a weak sense and $\dad_i u \in \L^2(V)$, with the obvious inner product and norm. Again we define $\Ht_0^1(V)$ to be the completion of $C_c^\infty(V)$ with respect to the norm $\norm{\cdot}{\Ht^1(V)}$.

Let us state some properties of functions in these spaces.

\begin{Lemma} \label{H1lem}
\begin{enumerate}[(i)]
\item Functions of the form $u = \rho^{-\alpha} v$, with $v \in C^\infty(\Ub)$ are dense in $\H^1(U)$.
\item Suppose $u \in \H^1_0(U)$, then $\partial^\beta_i u \in \L^2(U)$ for any $\beta$.
\item Suppose $u \in \H^1(U)$. Then in a collar neighbourhood $[0, \epsilon) \times \dU\subset U$ of the boundary, we have $u \in C^0((0, \epsilon); L^2(\dU))$, with the expansion
\be
u = \rho^{-\alpha}(u_0 + \O{\rho^\alpha})
\ee
where $u_0 \in L^2(\dU)$, with $u_0 = 0$ iff $u \in \H^1_0(U)$. Furthermore, for any $\delta>0$, there exists a $C_\delta$ such that
\ben{trest}
\norm{u_0}{L^2(\dU)} \leq \delta \norm{u}{\H^1(U)} + C_\delta \norm{u}{\L^2(U)}
\een
\end{enumerate}
Similar results hold for $\Ht$, but with $\alpha$ replaced by $1-\alpha$.
\end{Lemma}
Part $(i)$ follows from a result of Kufner \cite{Kufner}, and parts $(ii)$-$(iii)$ may be derived by showing that the inequalities hold on suitable dense subsets. From this we see that if $u \in \H^1(U)$, then $u$ may `blow up like $\rho^{-\alpha}$ near $\dU$', whereas if $u \in \H^1_0(U)$ then $u$ is `bounded near $\dU$' in some appropriate sense. These spaces thus capture, to a certain degree, the boundary behaviour we hope for in our solutions. A consequence of the proof of $(ii)$ is that $\H^1_0(U)=\Ht^1_0(U)$.

In fact, we can prove a sharper result about the range of the trace operator, together with an extension result:
\begin{Lemma}\label{tracelem}
The operator $T\circ \rho^\alpha$, where $T$ is the trace operator, maps $\H^1(U)$ into $H^{\alpha}(\dU)$, and the map is surjective. Furthermore there exists a bounded right inverse so that corresponding to any $u_0 \in H^{\alpha}(\dU)$, there exists a $u\in \H^1(U)$ with $\rho^\alpha u|_{\dU} = u_0$ in the trace sense, with the estimate
\be
\norm{u}{\H^1(U)} \leq C \norm{u_0}{H^\alpha(\dU)}
\ee
where $C$ is independent of $u_0$.
\end{Lemma}
This follows from the fact that $\rho^\alpha u$ belongs to a weighted Sobolev space, to which one may apply the results of \cite{Kim}.

We will also require the spaces $\H_0^{1}(U)^*$ and $\H^1(U)^*$, the dual spaces of $\H_0^1(U)$ and $\H^1(U)$ respectively. If $f \in X^*$, $u \in X$ we denote the pairing by
\be
\pair{f}{u},
\ee
and define
\be
\norm{f}{X^*} = \sup \{\pair{f}{u}|u\in X, \, \norm{u}{X}\leq 1 \}.
\ee

It will be convenient, for notational compactness, to define the following spaces and norms. The Neumann data space $H^1_\textrm{data, N}(V)$ consists of triples $(u_0, u_1, \f)$ with $u_0 \in  \H^1(V)$, $u_1 \in \L^2(V)$ and $\f \in L^2([0, T]; \L^2(V))$, whereas for the Dirichlet data space $H^1_\textrm{data, D}(V)$ we additionally require $u_0 \in \H^1_0$. For both spaces, we define
\be
\norm{(u_0, u_1, \f)}{H^1_\textrm{data}(V)}^2 = \norm{u_0}{\H^1(V)}^2+ \norm{u_1}{\L^2(V)}^2+\norm{\f}{L^2([0, T]; \L^2(V))}^2,
\ee 
We take $H^1_\textrm{sol., D}(V)$ to consist of $\u\in L^\infty([0, T]; \H^1_0(V))$ with 
\be
\norm{\u}{H^1_\textrm{sol., D}(V)}^2 = \norm{\u}{L^\infty([0, T]; \H^1(V))}^2 + \norm{\u}{W^{1, \infty}([0, T]; \L^2(V))}^2+\norm{\u}{H^2([0, T]; (\H^1_0(V))^*)}^2<\infty,
\ee
and $H^1_\textrm{sol., N}(V)$ to consist of $\u$ with
\be
\norm{\u}{H^1_\textrm{sol., N}(V)}^2 = \norm{\u}{L^\infty([0, T]; \H^1(V))}^2 + \norm{\u}{W^{1, \infty}([0, T]; \L^2(V))}^2+\norm{\u}{H^2([0, T]; (\H^1(V))^*)}^2<\infty.
\ee

\subsection{The Weak Formulations}

In order to motivate the definition of the weak solutions, let us suppose that we have a solution to 
\ben{WP11}
 u_{tt} +\mathcal{L} u= f\qquad \textrm{ on } U
\een
which is sufficiently smooth for the following operations to make sense. We can multiply the equation by a smooth function $v$, integrate over $U$ and integrate by parts to establish
\be
\int_U \left( u_{tt} v + a_{ij} \da_i u \da_j v + b_i \da_i u v + c u v \right)\, \rho dx =\int_U f v \, \rho dx +  \int_{\dU} (\rho^{1-\alpha} n_i a_{ij} \da_j u) (\rho^\alpha v) dS
\ee
The surface term will vanish either if $u$ satisfies the Neumann boundary conditions, or else if $v$ satisfies the Dirichlet conditions. We define the following bilinear form on $\H^1(V)$
\ben{WP12}
B_V[u,v;t] = \int_V \left[a_{ij}(\da_i u) (\da_j v) + b^i (\da_i u) v + c u v\right] \rho\ dx.
\een
If $B$ has no subscript, we assume the range to be $U$. Now we may define the weak Dirichlet and Neumann problems:
\begin{Definition}[Weak Dirichlet IBVP]
Suppose $(u_0, u_1, \f)\in H^1_\textrm{data, D}(U)$. We say that $\u \in H^1_\textrm{sol., D}(U)$ is a \emph{weak solution} of the Dirichlet IBVP if
\begin{enumerate}[i)]
\item For all $v \in \H^1_0(U)$ and a.e. time $0 \leq t \leq T$ we have
\be
\pair{\ddot{\u}}{v} + B[\u,v;t] = \ip{\mathbf{f}}{v}{\L^2(U)}.
\ee
\item We have the initial conditions
\be
\u(0) = u_0, \qquad \dot{\u}(0) = u_1.
\ee
\end{enumerate}
\end{Definition}

\begin{Definition}[Weak Neumann IBVP]
Suppose $(u_0, u_1, \f)\in H^1_\textrm{data, N}(U)$. We say that $\u \in H^1_\textrm{sol., N}(U)$ is a \emph{weak solution} of the Neumann IBVP if
\begin{enumerate}[i)]
\item For all $v \in \H^1(U)$ and a.e. time $0 \leq t \leq T$ we have
\be
\pair{\ddot{\u}}{v} + B[\u,v;t] = \ip{\mathbf{f}}{v}{\L^2(U)}.
\ee
\item We have the initial conditions
\be
\u(0) = u_0, \qquad \dot{\u}(0) = u_1.
\ee
\end{enumerate}
\end{Definition}

We note that by the calculation above, a strong solution obeying the Dirichlet (resp. Neumann) condition on the boundary is necessarily a weak Dirichlet (resp. Neumann) solution. The converse of course need not be true, however if we have enough regularity to integrate by parts then taking an arbitrary $v \in \H^1_0(U)$ we conclude that \eq{WP11} holds almost everywhere in $U$ in both the Dirichlet and Neumann case. Noting in the Neumann case that the trace of $\rho^\alpha v$ is arbitrary on the boundary we can, with care, deduce that $n_i a_{ij} \da_j u\in \tilde\H^1_0(U)$. We will see this in more detail later when we consider the asymptotics of the solutions.

\subsection{The theorems}

We're now ready to prove the well posedness of solutions to the weak formulations of \eq{WP1}. First, we have the following result

\begin{Theorem}[Uniqueness of weak solutions]  \label{Uniq}
Suppose $u$ is a weak solution of either the Dirichlet IBVP or of the Neumann IBVP. Then $u$ is unique. 
\end{Theorem}
\begin{proof}
The proof of uniqueness for the weak solutions proceeds almost identically to the proof of uniqueness of weak solutions to a finite IBVP. Without loss of generality, one may assume trivial data. In both cases one may take as test function
\ben{WP13}
\mathbf{v}(t) = \left \{ \begin{array}{lcl} \int_t^s \u(\tau) d\tau & \quad & 0\leq t \leq s \\ 0 &\quad & s\leq t \leq T \end{array} \right . ,
\een
and then integrate the weak equation over $0\leq t\leq s$. Standard manipulations making use of the uniform hyperbolicity condition \eq{WP5} then show $\u=0$. For example, one may take the proof of Evans \cite[p.~385]{Evans} and replace the standard spatial derivatives with twisted derivatives.
\end{proof}

\begin{Theorem}[Existence of weak solutions] \label{Exist}
\begin{enumerate}[(i)]
\item Given $(u_0, u_1, \f)\in H^1_\textrm{data, D}(U)$, there exists a weak solution to the Dirichlet IBVP corresponding to this data.
\item Given data $(u_0, u_1, \f)\in H^1_\textrm{data, N}(U)$, there exists a weak solution to the Neumann IBVP corresponding to this data.
\end{enumerate}
In both cases, we have the following estimate
\be
\norm{\u}{H^1_{\textrm{sol.,}\dagger}(U)} \leq C \norm{(u_0, u_1, \f)}{H^1_{\textrm{data}}(U)}
\ee
where $C$ depends on $T, U, \alpha$ and the coefficients of the equation. $\dagger$ stands for D or N as appropriate.
\end{Theorem}
It is convenient to divide the proof of Theorem \ref{Exist} into several Lemmas. We start by picking a sequence $0<a_{k}<\epsilon$ which decreases monotonically to zero, and define the sets $V_k = \{x:\rho(x)>a_k\}$. The broad strategy is to solve the finite IBVP on each $V_k$, where the equation becomes strictly hyperbolic and classical theory applies, and then find a way of passing to the limit $`k \to \infty$'.

\begin{Lemma} \label{exlem1} \begin{enumerate}[(i)]
\item Data $(u_0, u_1, f)$ such that for all  $k>k_0$ the problem
 \bea \nonumber
 &u_{tt} +\mathcal{L} u= f\textrm{ on } [0, T]\times V_k ,&\\
 & u = \left  . u_0\right |_{V_k}, u_t = \left . u_0\right |_{V_k} u = 0 \textrm{ on } \{0\} \times V_k, \qquad u=0  \textrm{ on } [0, T]\times  \partial V_k & \label{WP14}
 \eea
 has a solution which is $C^\infty([0, T]\times \overline{V}_k)$ form a dense linear subspace of $H^1_\textrm{data, D}(U)$
 \item Data $(u_0, u_1, f)$ such that for all $k>k_0$ the problem
 \bea \nonumber
 &u_{tt} +\mathcal{L} u= f\textrm{ on } [0, T]\times V_k ,&\\
 & u = \left  . u_0\right |_{V_k}, u_t = \left . u_0\right |_{V_k} u = 0 \textrm{ on } \{0\} \times V_k, \qquad n_i a_{ij} \da_j u=0  \textrm{ on } [0, T]\times  \partial V_k & \label{WP15}
 \eea
 has a solution which is $C^\infty([0, T]\times \overline{V}_k)$ form a dense linear subspace of  $H^1_\textrm{data, N}(U)$
 \end{enumerate}
\end{Lemma}
\begin{proof}
For $(i)$, we may take $u_0, u_1 \in C_c^\infty(U)$ and $\f \in C_c^\infty([0, T]\times U)$. For large enough $k$ the data are supported inside $V_k$, and data of this form a dense linear subspace of $H^1_\textrm{data, D}(U)$. For (ii), we need the fact that smooth functions $u$ for which $n_i  a_{ij} \da_j u =0 $ outside a compact set are dense in $\H^1(U)$. To see this, we first note that for any $u\in \H^1(U)$ we may take $u^{\epsilon} = \rho^{-\alpha} v^\epsilon$, where $v^\epsilon \in C^\infty (\Ub)$ and
\be
\norm{u-u^\epsilon}{\H^1(U)} < \epsilon.
\ee
We define $v^\epsilon_0$ in a collar neighbourhood of the boundary $[0, \delta)\times \dU$ to satisfy
\be
n_i a_{ij} \partial_j v^\epsilon_0 = 0, \qquad \left. v^\epsilon_0\right |_{\rho=0} = \left. v^\epsilon \right|_{\dU}.
\ee
Here $n_i$ is the unit normal of $\rho=const.$, which defines a smooth vector field provided $\delta$ is sufficiently small. Take $\chi(\rho)$ to be a smooth function, equal to $1$ for $\rho<\delta/2$ and vanishing for $\rho>3\delta/4$. Now, $u^\epsilon - \rho^{-\alpha} v^\epsilon_0\chi(\rho) =\tilde{u}^\epsilon \in \H^1_0(U)$, so there exists $w\in C^\infty_c(U)$ such that $\norm{\tilde{u}^\epsilon-w^\epsilon}{\H^1(U)}<\epsilon$. Consider the function $y^\epsilon=\rho^{-\alpha} v_0^\epsilon + w^\epsilon$. This satisfies
\be
\norm{u-y^\epsilon}{\H^1(U)}<2 \epsilon
\ee
and
\be
n_i a_{ij} \da_j  y^\epsilon = 0,\quad \textrm{ near }\dU.
\ee
We can suppose then that $\rho^\alpha u_0\in C^\infty(\Ub)$ with $n_i  a_{ij} \da_j u_0 =0 $ near $\dU$, and take $u_1\in C_c^\infty(U)$. Finally we can take a smooth $\f$ which is a sum of one component in $C_c^\infty(U)$  and another of arbitrarily small $L^2([0, T]; \L^2(U))$ norm which ensures the higher order compatibility conditions vanish to all orders on $t=0$. For large enough $k$ this data will launch a smooth solution and such data are dense in $H^1_\textrm{data, N}(U)$.
\end{proof}

\begin{Lemma} \label{exlem2}
Suppose $u$ is a smooth solution of 
\bea \nonumber
 &u^k_{tt} +\mathcal{L} u^k= f\quad \textrm{ on } [0, T]\times V_k ,&\\
 & u^k = \left  . u_0\right |_{V_k},\ \  u^k_t = \left . u_1\right |_{V_k} \ \   \textrm{ on } \{0\} \times V_k,  & \label{WP16}
 \eea
 with either $u^k=0$ or $n_i a_{ij} \da_j u^k = 0$ on $\partial V_k$. Then $u^k$ satisfies the estimate
\ben{WP17}
\norm{\u^k}{H^1_{\textrm{sol.,}\dagger}(V_k)} \leq C \norm{(u_0, u_1, \f)}{H^1_\textrm{data}(U)}
\een
where $C$ is uniform in $k$. $\dagger$ stands for D or N as appropriate.
\end{Lemma}
\begin{proof}
We drop the superscript on the solutions $u^k$ for convenience. Multiplying by $u_t$ and integrating by parts, using the boundary condition to neglect the boundary term, one has
\ben{WP18}
\ip{\ddot{\u}}{\dot{\u}}{\L^2(V_k)} +B_{V_k}[\u, \dot{\u}; t] = \ip{\f}{\dot{\u}}{\L^2(V_k)}.
\een
We also note that
\ben{WP18.5}
\frac{d}{dt}\frac{1}{2}\norm{\u}{\L^2(U)}^2 = \ip{\u}{\dot{\u}}{\L^2(U)}
\een
Taking \eq{WP18} and adding it to $\gamma$ times \eq{WP18.5}, we arrive at the equality
\bea \nonumber
&& \frac{d}{dt}\frac{1}{2}\left[\norm{\dot{\u}}{\L^2(V_k)}^2 + B_{V_k}[\u, \u; t] + \gamma\norm{{\u}}{\L^2(V_k)}^2  \right] =  \quad \ip{\f}{\dot{\u}^k}{\L^2(V_k)}\\ &&  - \int_{V_k} \left( \dot{a}_{ij} (\da_i \u) (\da_j \u) + \dot{b}^i (\da_i \u) \u + b^i (\da_i \u) \dot{\u} + \dot{c} \u^2 + \gamma \u \dot{u} \right) \rho dv \label{WP18.55}
\eea

Note that we have a bound
\be
\sup_{U} \abs{a_{ij}},  \abs{\dot{a}_{ij}}, \abs{b^i}, \abs{\dot{b}^i}, \abs{c}, \abs{\dot{c}} <C
\ee
which together with the uniform hyperbolicity condition:
\be
\theta \abs{\xi}^2 \leq a_{ij}(x) \xi^i \xi^j, \qquad \textrm{ for all } x\in U,\ \xi \in \mathbb{R}^N
\ee
implies that there exist $\gamma, M$, independent of $k$ such that for each $t$
\ben{WP18.6}
\norm{\dot{\u}}{\L^2(V_k)}^2 + \norm{\u}{\H^1(V_k)}^2 \leq M \left(\norm{\dot{\u}}{\L^2(V_k)}^2 + B_{V_k}[\u, \u; t] + \gamma\norm{{\u}}{\L^2(V_k)}^2 \right)
\een
holds for all smooth $u$. To see this, recall from \eq{WP12} that
\be
B_{V_k}[\u,\u;t] = \int_{V_k} \left[a_{ij}(\da_i \u) (\da_j \u) + b^i (\da_i \u) v + c \u^2\right] \rho\ dx.
\ee
Applying the uniform hyperbolicity estimate to the first term on the right hand side and the Cauchy-Schwarz inequality to the second term, we have that for any $\delta>0$
\be
B_{V_k}[\u,\u;t] \geq  (\theta-\delta) \norm{\u}{\H^1(U)}^2 - C_\delta \norm{u}{\L^2(U)}^2
\ee
where
\be
C_\delta = \sup_{U}\left( \abs{c} + \frac{\abs{b}^2}{4\delta}+\theta\right).
\ee
Taking $\delta = \theta/2$ and $\gamma > C_\delta$, we conclude that \eq{WP18.6} holds with $M = 1+ 2/\theta$.

We can now estimate from \eq{WP18.55}, \eq{WP18.6} and making use of the fact that we have  bounds on the coefficients which are uniform in $k$:
\bea \nonumber
&& \frac{d}{dt}\left[\norm{\dot{\u}}{\L^2(V_k)}^2 + B_{V_k}[\u, \u; t] + \gamma\norm{{\u}}{\L^2(V_k)}^2  \right] \leq  \\ &&\qquad C\left[\norm{\f}{\L^2(V_k)}^2 + \norm{\dot{\u}}{\L^2(V_k)}^2 + B_{V_k}[\u, \u; t] + \gamma\norm{{\u}}{\L^2(V_k)}^2  \right] 
\eea
with $C$ independent of $k$. Using Gronwall's lemma, together with a further application of \eq{WP18.6} we arrive at \eq{WP17}.

\end{proof}

\begin{Lemma}[Weak compactness] \label{exlem3}
\begin{enumerate}[(i)]
\item Suppose $\u^k \in H^1_{\mathrm{sol. D}}(V_k)$, with
\be
\norm{\u^k }{ H^1_{\mathrm{sol. D}}(V_k)} \leq C
\ee
Then there exists $\u \in H^1_{\mathrm{sol. D}}(U)$ with $\norm{\u }{ H^1_{\mathrm{sol. D}}(U)} \leq C$ and a subsequence $\u^{k_l}$ such that for any $v \in \H^1_0(V_m)$, taking $l$ large enough that $k_l>m$ we have for almost every $t$:
\bea
\ip{\u^{k_l}(t)}{v}{\L^2(V_{k_l})} &\to& \ip{\u(t)}{v}{\L^2(U)}\nonumber \\
\ip{\da_i \u^{k_l}(t)}{v}{\L^2(V_{k_l})} &\to& \ip{\da_i \u(t)}{v}{\L^2(U)}\nonumber \\
\ip{\dot{\u}^{k_l}(t)}{v}{\L^2(V_{k_l})} &\to& \ip{\dot{\u}(t)}{v}{\L^2(U)} \\
\pair{\ddot{\u}^{k_l}(t)}{v} & \to & \pair{\ddot{\u}(t)}{v}, \nonumber
\eea
\item Suppose $\u^k \in H^1_{\mathrm{sol. N}}(V_k)$, with
\be
\norm{\u^k }{ H^1_{\mathrm{sol. N}}(V_k)} \leq C
\ee
Then there exists $\u \in H^1_{\mathrm{sol. N}}(U)$ with $\norm{\u }{ H^1_{\mathrm{sol. N}}(U)} \leq C$ and a subsequence $\u^{k_l}$ such that for any $v \in \H^1(U)$, we have for almost every $t$:
\bea
\ip{\u^{k_l}(t)}{\left . v\right |_{V_{k_l}}}{\L^2(V_{k_l})} &\to& \ip{\u(t)}{v}{\L^2(U)}  \nonumber \\
\ip{\da_i \u^{k_l}(t)}{\left . v\right |_{V_{k_l}}}{\L^2(V_{k_l})} &\to& \ip{\da_i \u(t)}{v}{\L^2(U)}\nonumber \\
\ip{\dot{\u}^{k_l}(t)}{\left . v\right |_{V_{k_l}}}{\L^2(V_{k_l})} &\to& \ip{\dot{\u}(t)}{v}{\L^2(U)} \\
\pair{\ddot{\u}^{k_l}(t)}{\left . v\right |_{V_{k_l}}} & \to & \pair{\ddot{\u}(t)}{v}, \nonumber
\eea
\end{enumerate}
\end{Lemma}
\begin{proof}
We demonstrate first the proof for $u^k \in \H^1(V_k)$, $\norm{u}{\H^1(V_k)}\leq C$, i.e. the first part of (ii). We define $\overline{u^k}\in\L^2(U)$ to agree with $u^k$ on $V_k$ and to vanish on $U \setminus V_k$. Similarly, we define $\overline{\da_i u^k}\in \L^2(U)$ to agree with $\da_i u^k$ on $V_k$ and to vanish on $U \setminus V_k$. Weak compactness of $\L^2(U)$ gives a weakly convergent subsequence $(\overline{u^{k_l}}, \overline{\da_i u^{k_l}}) \wto (u, v_i)$. It remains to show that $v_i=\da_i u$ in the weak sense. To show this, multiply $\overline{\da_i u^{k_l}}$ by $\phi \in C^\infty_c(U)$ and integrate over $U$. For $l$ large enough that $\textrm{supp } \phi \subset V_{k_l}$ , we have $ \int_U \phi \overline{\da_i u^{k_l}} \rho dx = -\int_U \dad_i \phi \overline{u^{k_l}} \rho dx $, so by taking weak limits we're done. Similar considerations may be applied to the other results in the Lemma, after applying Riesz representation theorem to $\ddot{\u}$.
\end{proof}
\textbf{Remark :} This Lemma can be extended to apply to higher spatial derivatives of $u$, in an essentially unchanged fashion.

Now we can combine the results above to show that there exists a solution to the weak problems.

\begin{proof}[Proof of Theorem \ref{Exist}]
\begin{enumerate}[(i)]
\item Suppose we have data $(u_0, u_1, f)$ such that for all  $k>k_0$ the problem
 \bea \nonumber
 &u_{tt} +\mathcal{L} u= f\textrm{ on } [0, T]\times V_k ,&\\
 & u = \left  . u_0\right |_{V_k}, u_t = \left . u_0\right |_{V_k} u = 0 \textrm{ on } \{0\} \times V_k, \qquad u=0  \textrm{ on } [0, T]\times  \partial V_k &
 \eea
 has a solution, $u^k$ which is $C^\infty([0, T]\times \overline{V}_k)$. By Lemma \ref{exlem2}, we have the estimate
\be
\norm{\u^k}{H^1_{\textrm{sol.,}D}(V_k)} \leq C \norm{(u_0, u_1, \f)}{H^1_\textrm{data}(U)}
\ee
And we also know that for $k>m$ and for any $v \in \H_0^1(V_{m})$ we have
\ben{WP22}
\pair{u^k_{tt}}{v}+ B_{V_k}[u^k, v; t] = \ip{f}{v}{\L^2(V_k)}
\een
Applying Lemma \ref{exlem3} we conclude the existence of $u$ satisfying
\ben{WP23}
\norm{\u}{H^1_{\textrm{sol.,}D}(U)} \leq C \norm{(u_0, u_1, \f)}{H^1_\textrm{data}(U)}
\een
and for any $v \in \H_0^1(V_{m})$:
\ben{WP24}
\pair{u_{tt}}{v} + B[u, v; t] = \ip{f}{v}{\L^2(U)}.
\een
Noting that functions $v \in \H_0^1(V_{m})$ are dense in $\H^1_0(U)$, we conclude that $u$ satisfies the first condition to be a weak solution of the Dirichlet IBVP. We must now check that the weak solution we have constructed satisfies the initial conditions. For this, choose any function $\v \in C^2(0, T; C^\infty_c(U))$, with $\v(T)=\dot{\v}(T)=0$. Integrating \eq{WP24} in time, we have after twice integrating by parts
\be
\int_0^T\pair{\ddot{v}}{\u} + B[\u,\v;t] dt = \int_0^T\ip{\mathbf{f}}{\v}{\L^2(U)} dt-\pair{\u(0)}{\dot{\v}(0)}+\pair{\dot{\u}(0)}{\v(0)}
\ee
similarly, we have from \eq{WP22}
\be
\int_0^T\pair{\ddot{v}}{\u^k}_{V_l} + B_{V_k}[\u^k,\v; t] dt = \int_0^T\ip{\mathbf{f}}{\v}{\L^2(V_l)} dt-\pair{\u^k(0)}{\dot{\v}(0)}_{V_k}+\pair{\dot{\u}^k(0)}{\v(0)}_{V_k}.
\ee
Setting $k = k_l$, passing to the limit we have:
\be
\int_0^T\pair{\ddot{v}}{\u} + B[\u,\v; t] dt = \int_0^T\ip{\mathbf{f}}{\v}{\L^2(U)} dt-\pair{u_0}{\dot{\v}(0)}+\pair{u_1}{\v(0)}.
\ee
Since $\v(0), \dot{\v}(0)$ are arbitrary, we conclude that $\u(0)=u_0$, $\dot{\u}(0)=u_1$ and we're done. 

Finally, we make use of Lemma \ref{exlem1} together with the uniqueness result Theorem \ref{Uniq} and a standard argument based on continuity, using \eq{WP23}, to show that our result holds for any $(u_0, u_1, \f) \in H^1_{\textrm{data, D}}(U)$.
 
\item The Neumann case follows in an almost identical manner, solving a sequence of finite Neumann problems for suitably smooth data and using the weak compactness to extract a weak solution. Finally a continuity argument again extends the existence proof to all admissible data.
\end{enumerate}
\end{proof}

\section{Higher Regularity and Asymptotics\label{HR}}

We now wish to show that if more assumptions are made on the data, the weak solution can be shown to have improved regularity. In order to do this, we will require some elliptic estimates, enabling us to control some appropriate $H^2$ norm of $u$ in terms of $\mathcal{L} u$. 

\subsection{The $\H^2$ norm}

The $\H^2$ norm is slightly unusual in its definition because it is necessary to distinguish the directions tangent to the boundary from those normal to it. We fix a finite set of vector fields $\{T^{(A)}, N^{(B)}\}$ on $U$ which satisfy the following properties:
\begin{enumerate}[(i)]
\item For $\rho < \tilde \epsilon$, we have $T^{(A)}$ normal to $\partial_i \rho$, while $N^{(B)}$ are parallel\footnote{recall $n_i \dot{a}_{ij}=0$ on $\dU$, so the spaces we construct are equivalent for any value of $t$} to $a_{ij} \partial_j \rho$.
\item At each point of $\Ub$, the set $\{T^{(A)}, N^{(B)}\}$ spans $\mathbb{R}^N$.
\end{enumerate}
\begin{Definition}
We say that a function $u\in \H^1(U)$ belongs to $\H^2(U)$, provided
\be
T^{(A)}_i \da_i u \in \H^1(U), \qquad N^{(B)}_i \da_i u \in \Ht^1(U),
\ee
for all $A, B$, and we define the norm:
\be
\norm{u}{\H^2(U)}^2 = \norm{u}{\H^1(U)}^2 + \sum_A \norm{T^{(A)}_i \da_i u }{ \H^1(U)}^2+\sum_B \norm{N^{(B)}_i \da_i u}{\Ht^1(U)}^2
\ee
\end{Definition}
\textbf{Remarks :} \begin{enumerate}[(a)] 
\item A different choice of  $\{T^{(A)}, N^{(B)}\}$ satisfying (i), (ii) gives rise to an equivalent norm. \item If $u\in \H^2(U)$ then $u \in H^2_{\textrm{loc.}}(U)$.
\item If $u \in \H^2(U)$, then $\dad_i\left(\frac{\partial^m a_{ij}}{\partial t^m} \da_j u \right) \in \L^2(U)$. This observation is important in establishing the higher regularity energy estimates.
\end{enumerate}

\subsection{Elliptic estimates}

We first define the weak version of the elliptic problem we study. We assume that $t$ is fixed throughout this section:
\begin{Definition}
Suppose $f \in \H^1_0(U)^*$ (resp. $\H^1(U)^*$). We say that $u \in \H^1_0(U)$ (resp. $\H^1(U)$) is a weak solution of the Dirichlet (resp. Neumann) problem
\ben{EL1}
\mathcal{L} u = f  \quad \textrm{ on } U
\een
with $u=0$ (resp. $n_i a_{ij} \da_j u=0$) on $\dU$, if
\be
B[u, v] = \pair{f}{v}
\ee
for all $v \in \H^1_0(U)$ (resp. $\H^1(U)$).
\end{Definition}

\begin{Theorem}[Elliptic Estimates]\label{ellipticest}
Suppose $u$ is a weak solution of either the Dirichlet or Neumann problem \eq{EL1} and suppose that in fact $f \in \L^2(U)$. Then $u\in \H^2(U)$
with the estimate
\ben{EL2}
\norm{ u }{ \H^2(U)} \leq C \left ( \norm{f}{\L^2(U)}+ \norm{u}{\L^2(U)} \right)
\een
Furthermore, in the Dirichlet case $T_i \da_i u \in \H_0^1(U)$ and in the Neumann case $N_i\da_i \in \Ht^1_0(U)$.
\end{Theorem}
We split the result into several Lemmas
\begin{Lemma} \label{ELLem1}
There exist constants $C_1, C_2$ and $\mu_0 \geq 0$ such that
\begin{enumerate}[(i)]
\item $\abs{B[u,v]} \leq C_1 \norm{u}{\H^1(U)}^2 $
\item $ C_2 \norm{u}{\H^1(U)}^2 \leq B[u,u] +\mu_0 \norm{u}{\L^2(U)}^2$
\end{enumerate}
\end{Lemma}
\begin{proof} This is a standard manipulation, making use of the uniform ellipticity of $a_{ij}$.
\end{proof}
\begin{Lemma}\label{ELLem2}
There exists $\mu_0\in \mathbb{R}$ such that for all $\mu > \mu_0$, $f \in \L^2(U)$ the equation
\be
\mathcal{L}u + \mu u = f, \quad \textrm{ on } U
\ee
with either Dirichlet or Neumann boundary conditions, has a unique weak solution satisfying
\be
\norm{u}{\H^1(U)} \leq C \norm{f}{\L^2(U)}.
\ee
\end{Lemma}
\begin{proof}
Because of the estimates in the previous lemma, we may apply the Lax-Milgram theorem to $B_U[u, v]+\mu \ip{u}{v}{\L^2(U)}$ thought of as a bilinear form on either $\H_0^1(U)$ or $\H^1(U)$ for Dirichlet, Neumann conditions respectively.
\end{proof}
\begin{Lemma} \label{ELLem3}
Suppose $f \in C^\infty_c(U)$, and $k$ is sufficiently large that $\textrm{supp }f \subset V_k$. Then for $\mu > \mu_0$
\ben{EL3}
\mathcal{L}u + \mu u = f,  \quad \textrm{ on } V_k
\een
has a unique solution in $C^\infty(\overline{V_k})$. Furthermore, this solution obeys the estimate
\ben{EL4}
\norm{u}{\H^2(V_k)} \leq C \norm{f}{\L^2(U)}
\een
with $C$ uniform in $k$.
\end{Lemma}
\begin{proof}
We can apply Lemmas \ref{ELLem1}, \ref{ELLem2} on $V_k$ to deduce the existence of a unique weak solution to \eq{EL3} with the appropriate boundary conditions, satisfying
\be
\norm{u}{\H^1(V_k)} \leq C \norm{f}{\L^2(U)}.
\ee
where $C$ is independent of $k$. On $V_k$, the operator $\mathcal{L}$ is uniformly elliptic in the standard sense, so classical elliptic estimates imply that $u$ is smooth. We have that
\be
\dad_i(a_{ij} \da_j u) = \tilde{f}
\ee
where $\norm{\tilde{f}}{\L^2(V_k)} \leq C \norm{f}{\L^2(U)}$. Focusing on a coordinate patch, we can work assuming $V_k=\{(x, x^a): x>\epsilon_k,\ x^2+x^a x^a < \tilde\epsilon\}$, $\rho=x$ and assume that $\zeta$ is a smooth cut-off function on $\{(x, x^a): x\geq0,\ x^2+x^a x^a < \tilde\epsilon\}$ which vanishes on the curved part of the boundary. We note that $\da_a = \partial_a$ and that $\da_i \partial_a =\partial_a \da_i$. Now consider the following integral, where the index $A\in \{2, 3, \ldots, N\}$ is a \emph{fixed} index, with no summation over it.
\bea\label{EL5}
I&:=&\int_{V_k} [\dad_i(a_{ij} \da_j u)] [ \partial_A(\zeta^2 \partial_A u)] x dx \\\nonumber &\leq& \delta \norm{\zeta^2 \partial_A \partial_A u}{\L^2(V)}+C\left (\norm{f}{\L^2(V)}^2+\norm{u}{\H^1(V)}^2\right ).
\eea
By choosing $C$ large enough, we may take $\delta$ to be arbitrarily small.
Now, since $\zeta$ vanishes on the curved part of $\partial V$ and either $u$ or $n_i a_{ij} \da_j u$ vanishes on the flat part, we may integrate by parts twice to find
\bea
I &=& \int_{V_k} [\partial_A (a_{ij} \da_j u)][\da_i(\zeta^2 \partial_A u)] x dx\nonumber \\ &\geq& \int_V \zeta^2 a_{ij} (\da_i \partial_A u)(\da_j \partial_A u) x dx  - C\norm{u}{\H^1(V)}^2 \label{EL6}\\ &\geq& \theta \norm{\zeta^2 \da_i \partial_A u}{\L^2(V)} - C\norm{u}{\H^1(V)}^2, \nonumber 
\eea
where in the last line we have used the uniform ellipticity of $a_{ij}$. The constant $C$ here depends on the functions $a_{ij}, \zeta$, which are uniformly bounded in $k$. Now taking \eq{EL4}, \eq{EL5}, \eq{EL6} together, and choosing $\delta$ sufficiently small, we have that
\ben{EL6.5}
\norm{\zeta^2 \partial_A u}{\H^1(V_k)} \leq C \norm{f}{\L^2(U)}
\een 
with $C$ uniform in $k$, so we have estimated the tangential derivatives. Returning now to the equation, we can write
\be
\dad_x (a_{xi}\da_i u) = \tilde{f} - \partial_a (a_{ai}\da_i u) 
\ee
Multiplying by $\zeta^2$, we readily estimate
\ben{EL7}
\norm{\zeta^2 a_{xi} \da_i u}{\tilde{\H}^1(V_k)} \leq C  \norm{f}{\L^2(U)}.
\een
Combining these estimates with a partition of unity subordinate to a set of coordinate patches covering the boundary and an interior estimate which follows from standard elliptic theory, we're done.
\end{proof}
\begin{proof}[Proof of Theorem \ref{ellipticest}]
First we note that if $u$ is a weak solution of \eq{EL1} with either Dirichlet or Neumann boundary conditions, then $u$ is the unique weak solution of
\ben{EL8}
\mathcal{L}u + \mu u = \tilde{f}
\een
with $\tilde{f}=f+\mu u$ for sufficiently large $\mu$. Suppose $\tilde{f}\in C^\infty_c(U)$. Then we can solve the finite problems on $V_k$, with the estimate
\be
\norm{u^k}{\H^2(V_k)}\leq C \norm{\tilde{f}}{\L^2(U)}.
\ee
with $C$ uniform in $k$. We deduce the existence of a subsequence which tends weakly to $u$ in $\H^2$, as in the remark after Lemma \ref{exlem3}, so we find that
\be
\norm{u}{\H^2(U)} \leq C \norm{\tilde{f}}{\L^2(U)}
\ee
now, we may relax the condition that $\tilde f \in C^\infty_c(U)$ since such functions are dense in $\L^2(U)$. Finally, replacing $\tilde{f}$, we deduce
\be
\norm{u}{\H^2(V_k)} \leq C\left(  \norm{f}{\L^2(U)}+\norm{u}{\L^2(U)}\right).
\ee
Finally, note that for Dirichlet conditions we have $u^k \in \H^1_0(V_k)$ and for Neumann we have $n_i a_{ij} \da_j u^k \in \H^1_0(V_k)$, so that in the limit $u \in \H^1_0(U)$ or $n_i a_{ij} \da_j u \in \H^1_0(U)$ respectively.
\end{proof}

We would like to also prove elliptic estimates at a higher level of regularity than $H^2$. Unfortunately, the behaviour of the solutions near the boundary doesn't lend itself to a description in terms of a global Sobolev space. Accordingly then, we first consider interior regularity.
\begin{Theorem}
Suppose $u$ is a weak solution of either the Dirichlet or Neumann problem \eq{EL1} and suppose that in fact $f \in \L^2(U)\cap H^m_{\textrm{loc.}}(U)$. Then $u\in \H^2(U) \cap H^{m+2}_{\textrm{loc.}}(U)$
\end{Theorem}
\begin{proof}
This follows from standard elliptic estimates and the fact that $\mathcal{L}$ is uniformly elliptic on any $V \cc U$.
\end{proof}

To say more about the behaviour near the boundary, we shall once again need to distinguish directions tangent and normal to the boundary. It's convenient to introduce the space $\H_T^m(U)$, consisting of all functions $u$ such that
\be
T^{(1)} T^{(2)}\ldots T^{(l)} u \in \L^2(U)
\ee
for any $l\leq m$ smooth vector fields $T^{(i)}$ tangent to the boundary $\dU$. To capture the behaviour normal to the boundary, we work with an asymptotic expansion.

\begin{Theorem}\label{asympt}
Suppose $u$ is a weak solution of either the Dirichlet or Neumann problem \eq{EL1} where  $f \in \H_T^m(U)\cap H^m_{\textrm{loc.}}(U)$, where $m \geq 0$. Suppose further that if $m \geq 1$ near $\dU$ we have the following expansion for $f$:
\bea
f &=& \rho^{\alpha-1} \left[ f^+_0 + \rho f^+_1 + \ldots + \rho^{m-1} f^+_{m-1} + \O{\rho^{m-\alpha}})\right ] \\
&& \quad + \rho^{-\alpha} \left[ f^-_0 + \rho f^-_1 + \ldots + \rho^{m-1} f^-_{m-1} + \O{\rho^{m-1+\alpha}})\right ] \nonumber
\eea
where
\be
f^\pm_i \in H^{m-1-i}(\dU)
\ee
and $\norm{\rho^{-a} \O{\rho^a}}{L^2(\dU)}$ is bounded as $\rho \to 0$. Then $u\in \H^2(U) \cap \H^{m+2}_T (U)\cap H^{m+2}_{\textrm{loc.}}(U)$ has the following expansion for $m \geq 0$:
\bea\label{EL8.1}
u &=& \rho^{\alpha} \left[ u^+_1 + \rho u^+_2 + \ldots + \rho^{m} u^+_{m+1} + \O{\rho^{m+1-\alpha}})\right ] \\
&& \quad + \rho^{-\alpha} \left[ u^-_0 + \rho u^-_1 + \ldots + \rho^{m+1} u^-_{m+1} + \O{\rho^{m+1+\alpha}})\right ] \nonumber
\eea
where
\be
u^\pm_i \in H^{m+1-i}(\dU).
\ee
Furthermore if $u$ satisfies the Dirichlet conditions, $u^-_0=u^-_1=0$, while if $u$ satisfies the Neumann conditions $u^+_1=0$.
\end{Theorem}
\begin{proof}
The proof is by induction. To establish the $m=0$ case, we apply Lemma \ref{H1lem} and Theorem \ref{ellipticest} to deduce that in a coordinate patch near the boundary
\bea\nonumber 
a_{xi} \da_iu &=& x^{-\alpha}[\O{x^\alpha}] + x^{\alpha-1}[c_+ + \O{x^{1-\alpha}}],\\ \label{EL8.2} \partial_au &=& x^{-\alpha}[c_- +\O{x^\alpha}] + x^{\alpha-1}[\O{x^{1-\alpha}}].
\eea
with $c_\pm \in L^2(\dU)$. We thus have that
\be
\da_x u = x^{-\alpha}[\tilde c_-+ \O{x^\alpha}] + x^{\alpha-1}[\tilde c_+ + \O{x^{1-\alpha}}],
\ee
and integrating this gives \eq{EL8.1} for $m=0$, with $u^+_1, u^-_0, u^-_1 \in L^2(\dU)$. Finally we note that the second identity of \eq{EL8.2} implies $u^-_0 \in H^1(\dU)$. In order to get the induction step, we first commute with a vector field tangent to the boundary, which establishes all but the highest order in $\rho$ of \eq{EL8.1} by the induction assumption. To get the highest order terms, we re-arrange the equation $\mathcal{L}u = f$ to give an equation for $\dad_x \da_x u$, making use of the induction assumptions and integrate twice. Taking care of the boundary conditions imposed shows that for Dirichlet conditions, we have $u^-_0=u^-_1=0$, while for Neumann $u^+_1=0$.
\end{proof}
Taking a little more care about the origin of terms in the series, we can easily show
\begin{Corollary}
\begin{enumerate}[(i)]
\item If $u$, $f$ satisfy the conditions for Theorem \ref{asympt} with Dirichlet boundary conditions and furthermore $f^-_i = 0$ for $0 \leq i \leq m-1$, then $u^-_i = 0$ for $0 \leq i \leq m+1$.
\item If $u$, $f$ satisfy the conditions for Theorem \ref{asympt} with Neumann boundary conditions and furthermore $f^+_i = 0$ for $0 \leq i \leq m-1$, then $u^+_i = 0$ for $1 \leq i \leq m+1$.
\end{enumerate}
\end{Corollary}

\subsection{Higher regularity}

We define the higher regularity data spaces inductively as follows. We say $(u_0, u_1, \f) \in H^2_{\textrm{data, }D}$ if $u_0 \in \H^2(U), u_1 \in \H^1_0(U), \f \in H^1([0, T]; \L^2(U))$ with the product norm. In the Neumann case, $(u_0, u_1, \f) \in H^2_{\textrm{data, }N}$ if $u_0 \in \H^2(U), a_{ij}\da_j u_0 \in \H^1_0(U), u_1 \in \H^1(U), \f \in H^1([0, T]; \L^2(U))$. Next, we define
\bea
g_0 &=& u_0 \qquad g_1 = u_1 \nonumber \\
g_{i+2} &=& -\sum_{l=1}^i \left(\begin{array}{c} i \\ l \end{array} \right) \left(\mathcal{L}^{(l)} g_{i-l} \right)+ f^{(i)}|_{t=0}.
\eea
Here $\mathcal{L}^{(i)}$ is the second order operator given by differentiating the coefficients of $\mathcal{L}$ $i$ times with respect to $t$. For $m>2$, we say $(u_0, u_1, \f) \in H^m_{\textrm{data, }\dagger}$ if $(u_0, u_1, \f)  \in H^{m-1}_{\textrm{data, }\dagger}$, $\f^{(i)} \in L^2([0, T]; H^{m-i-1}_{\textrm{loc.}}(U))$ for $0 \leq i \leq m-1$ and $(g_{m-1}, g_m, \f^{(m-1)})\in H^1_{\textrm{data, }\dagger}$. Here as usual $\dagger$ stands for D or N as appropriate. We define the norms
\be
\norm{(u_0, u_1, \f)}{H^{m}_{\textrm{data, }\dagger}}^2 = \norm{(u_0, u_1, \f)}{H^{m-1}_{\textrm{data, }\dagger}}^2+\norm{(g_{m-1}, g_m, \f^{(m-1)})}{H^1_{\textrm{data, }\dagger}}^2
\ee
these spaces are chosen so that the relevant `compatibility conditions' hold. We may show that if $(u_0, u_1, \f) \in H^k_{\textrm{data, }\dagger}$, then $u_0 \in H^k_{\textrm{loc.}}(U)$, $u_1 \in H^{k-1}_{\textrm{loc.}}(U)$.

\begin{Theorem}
\begin{enumerate}[(i)]
\item Suppose $\u$ is a weak solution of the Dirichlet IBVP corresponding to data $(u_0, u_1, \f)\in H^1_\textrm{data, D}(U)$. Suppose in addition,
$(u_0, u_1, \f)\in H^2_\textrm{data, D}(U)$ then
\be
\begin{array}{l}
\u \in L^\infty(0, T; \H^2(U)), \quad \dot{\u}\in L^\infty(0, T; \H_0^1(U)),\\
\ddot{\u} \in L^\infty(0, T; \L^2(U)), \quad \dddot{\u} \in L^2(0, T; (\H^1_0(U))^*),
\end{array}
\ee
with the estimate
\bea\label{HR1}
&&\quad \esssup_{0\leq t\leq T}\, \left (\norm{\u(t)}{\H^2(U)} + \norm{\dot{\u}(t)}{\H^1(U)} + \norm{\ddot{\u}(t)}{\L^2(U)} \right)  \\
&&+\norm{\dddot{\u}}{L^2(0, T;(\H^1_0(U))^*)} \leq C\norm{(u_0, u_1, \f)}{H^2_\textrm{data, D}(U)}.\nonumber
\eea

\item Suppose $\u$ is a weak solution of the Neumann IBVP corresponding to data $(u_0, u_1, \f)\in H^1_\textrm{data, N}(U)$. Suppose in addition,
$(u_0, u_1, \f)\in H^2_\textrm{data, N}(U)$ then
\be
\begin{array}{l}
\u \in L^\infty(0, T; \H^2(U)), \quad \dot{\u}\in L^\infty(0, T; \H^1(U)),\\
\ddot{\u} \in L^\infty(0, T; \L^2(U)), \quad \dddot{\u} \in L^2(0, T; (\H^1(U))^*),
\end{array}
\ee
with the estimate
\bea\label{HR1.5}
&&\quad \esssup_{0\leq t\leq T}\, \left (\norm{\u(t)}{\H^2(U)} + \norm{\dot{\u}(t)}{\H^1(U)} + \norm{\ddot{\u}(t)}{\L^2(U)} \right)  \\
&&+\norm{\dddot{\u}}{L^2(0, T;(\H^1(U))^*)} \leq C\norm{(u_0, u_1, \f)}{H^2_\textrm{data, N}(U)}.\nonumber
\eea
Furthermore $n_i a_{ij} \da_j \u \in L^\infty([0,T]; \tilde{\H}^1_0(U))$.
\end{enumerate}
\end{Theorem}

\begin{proof}
First note that without loss of generality, we may take $u_0=0$, so that data which give rise to a smooth solution to the restricted problem on $V_k$ for $k$ sufficiently large are again dense. We return to the approximating sequence $\u^k$ we established in proving the existence of a weak solution. Commuting the equation with $\partial_t$ and making use of the elliptic estimates on $V_k$ established in the previous section it is straightforward to derive bounds for $\norm{\dot{\u}^k}{L^\infty([0, T]; \H^1(V_k))}$, $\norm{\ddot{\u}^k}{L^\infty([0, T]; \L^2(V_k))}$ and the relevant norm of $\dddot{\u}$ which are uniform in $k$. Passing to a weak limit and applying Theorem \ref{ellipticest} to deduce $\u \in L^\infty([0, T]; \H^2(U))$, we're done.
\end{proof}

Commuting further with $\partial_t$, it can be shown that the following theorem holds:
\begin{Theorem}[Higher Regularity] \label{HigherD}
\begin{enumerate}[(i)]
\item Assume $(u_0, u_1, \f) \in H^m_{\textrm{data,D}}$ and suppose also $\u$ is the weak solution of the Dirichlet IBVP problem with this data. Then in fact
\begin{eqnarray}
&& \esssup_{0 \leq t \leq T} \left(\sum_{i=0}^{m-2}\norm{\frac{d^i\u}{dt^i}}{\H^2(U)}+ \norm{\frac{d^{m-1}\u}{dt^{m-1}}}{\H^1(U)}+\norm{\frac{d^{m}\u}{dt^{m}}}{\L^2(U)} \right) \\ &&\qquad+\norm{\frac{d^{m+1}\u}{dt^{m+1}}}{L^2(0, T; (\H^{1}_0(U))^*)}\leq  \label{h1D} C\norm{(u_0, u_1, \f) }{ H^m_{\textrm{data,D}}}, 
\end{eqnarray}
where $C$ is a constant which depends on $T$ and $\alpha$ and the coefficients of the equation. Furthermore
\be
\u^{(i)}(t) \in H^{m-i}_{\textrm{loc.}}(U)\ \ \textrm{  for }\ \ 0 \leq i \leq k.
\ee
\item Assume $(u_0, u_1, \f) \in H^m_{\textrm{data,N}}$ and suppose also $\u$ is the weak solution of the Neumann IBVP problem with this data. Then in fact
\begin{eqnarray}
&& \esssup_{0 \leq t \leq T} \left(\sum_{i=0}^{m-2}\norm{\frac{d^i\u^l}{dt^i}}{\H^2(U)}+ \norm{\frac{d^{m-1}\u}{dt^{m-1}}}{\H^1(U)}+\norm{\frac{d^{m}\u}{dt^{m}}}{\L^2(U)} \right) \\ &&\qquad+\norm{\frac{d^{m+1}\u}{dt^{m+1}}}{L^2(0, T; (\H^{1}(U))^*)}\leq  \label{h1N} C\norm{(u_0, u_1, \f) }{ H^m_{\textrm{data,N}}}, 
\end{eqnarray}
where $C$ is a constant which depends on $T$ and $\alpha$ and the coefficients of the equation.
Furthermore
\be
\u^{(i)}(t) \in H^{m-i}_{\textrm{loc.}}(U)\ \ \textrm{  for }\ \ 0 \leq i \leq m.
\ee
\end{enumerate}
\end{Theorem}

Note that we do not directly control higher spatial derivatives of $\u$ in global Sobolev norms, although we do have control of powers of the $\mathcal{L}$ acting on $\u$. We can however make use of (a very slight adaptation of) Theorem \ref{asympt} to give the following asymptotic expansion in the situation where $\f$ has the appropriate behaviour near the boundary:

\begin{Theorem}\label{hypasympt}
Suppose $\u$ satisfies the conditions of Theorem \ref{HigherD} $(i)$ or $(ii)$. Suppose that $\f^{(i)} \in L^2([0, T]; \H_T^{m-i-1}(U)\cap H^{m-i-1}_{\textrm{loc.}}(U))$. Suppose further that if $m \geq 3$ near $\dU$ we have the following expansion for $f$:
\bea
f &=& \rho^{\alpha-1} \left[ f^+_0 + \rho f^+_1 + \ldots + \rho^{m-3} f^+_{m-3} + \O{\rho^{m-2-\alpha}})\right ] \\
&& \quad + \rho^{-\alpha} \left[ f^-_0 + \rho f^-_1 + \ldots + \rho^{m-3} f^-_{m-3} + \O{\rho^{m-3+\alpha}})\right ] \nonumber
\eea
where
\be
 f^\pm_i \in H^{m-3-i}(\dU_T)
\ee
Then $\u$ has the following expansion for $m \geq 1$:
\bea\label{HEL8.1}
u &=& \rho^{\alpha} \left[ u^+_1 + \rho u^+_2 + \ldots + \rho^{m-2} u^+_{m-1} + \O{\rho^{m-1-\alpha}})\right ] \\
&& \quad + \rho^{-\alpha} \left[ u^-_0 + \rho u^-_1 + \ldots + \rho^{m-1} u^-_{m-1} + \O{\rho^{m-1+\alpha}})\right ] \nonumber
\eea
where
\be
u^\pm_i \in H^{m-1-i}(\dU_T).
\ee
Furthermore if $u$ satisfies the Dirichlet conditions, $u^-_0=u^-_1=0$, while if $u$ satisfies the Neumann conditions $u^+_1=0$.
\end{Theorem}

For a formal power series approach to determining the coefficients of these expansions, see \cite{Gover:2011rz}.

\section{Other boundary conditions}\label{bcs}

We have now established well posedness and a regularity theory for solutions of  \eq{WP1} subject to either Dirichlet or Neumann homogeneous boundary conditions. We will discuss briefly some of the other possibilities listed in the introduction, although we shall not go into quite so much detail. 

\subsection{Inhomogeneous boundary data} \label{inhom}

First, we define the weak formulations for the inhomogeneous problems. We assume throughout this section that $u_0 \in \H^1(U), u_1 \in \L^2(U), \f \in L^2([0, T]; \L^2(U))$. We furthermore take $\mathbf{g_0}, \mathbf{g_1} \in L^\infty([0, T]; L^{2}(\dU))$ to be some functions on the boundary. The $\mathbf{g_i}$ will need to be subject to further conditions in order to give a well posed problem.
\begin{Definition}[Weak Inhomogeneous Dirichlet IBVP]
Suppose $u_0, u_1, \f$ are as above with the additional condition $\left. \rho^\alpha u_0 \right|_{\dU_T}=\mathbf{g_0}(0).$ We say that $\u \in L^\infty([0, T]; \H^1(U))$ with $\dot \u \in L^\infty([0, T]; \L^2(U))$, $ \ddot \u \in L^2([0, T]; (\H^1_0(U))^*)$ is a \emph{weak solution} of the inhomogeneous Dirichlet IBVP:
\be
\begin{array}{rcl}
u_{tt} + \mathcal{L} u = f & \quad & \textrm{ in } U_T \\
\left . \rho^\alpha u \right |_{\dU_T} = g_0 & \quad & \textrm{ on } \dU_T \\
u=u_0, {u_t} = u_1 & \quad & \textrm{ on } \{0 \} \times U
\end{array}
\ee
provided
\begin{enumerate}[i)]
\item For all $v \in \H^1_0(U)$ and a.e. time $0 \leq t \leq T$ we have
\be
\pair{\ddot{\u}}{v} + B[\u,v;t] = \ip{\mathbf{f}}{v}{\L^2(U)}.
\ee
\item We have the initial conditions
\be
\u(0) = u_0, \qquad \dot{\u}(0) = u_1.
\ee
\item We have the boundary condition
\be
\left . \rho^\alpha \u \right |_{\dU_T} = \mathbf{g_0}
\ee
\end{enumerate}
\end{Definition}
Note that in contrast to the homogeneous Dirichlet problem, the unrenormalized energy will be \emph{infinite} for a solution of the inhomogeneous Dirichlet problem with $g_0\neq 0$. 

\begin{Definition}[Weak Inhomogeneous Neumann IBVP]
Suppose $u_0, u_1, \f$ are as above. We say that $\u \in L^\infty([0, T]; \H^1(U))$ with $\dot \u \in L^\infty([0, T]; \L^2(U))$, $ \ddot \u \in L^2([0, T]; (\H^1(U))^*)$ is a \emph{weak solution} of the inhomogeneous Neumann IBVP:
\be
\begin{array}{rcl}
u_{tt} + \mathcal{L} u = f & \quad & \textrm{ in } U_T \\
\left . \rho^{1-\alpha}n_i a_{ij} \da_j u \right |_{\dU_T} = g_1 & \quad & \textrm{ on } \dU_T \\
u=u_0, {u_t} = u_1 & \quad & \textrm{ on } \{0 \} \times U
\end{array}
\ee
provided
\begin{enumerate}[i)]
\item For all $v \in \H^1(U)$ and a.e. time $0 \leq t \leq T$ we have
\be
\pair{\ddot{\u}}{v} + B[\u,v;t] = \ip{\mathbf{f}}{v}{\L^2(U)}+ \ip{\mathbf{ g_1}}{\left . \rho^\alpha v\right|_{\dU}}{L^2(\dU)}.
\ee
\item We have the initial conditions
\be
\u(0) = u_0, \qquad \dot{\u}(0) = u_1.
\ee
\end{enumerate}
\end{Definition}
If we assume sufficient regularity, it is possible to show that the weak solutions are equivalent to strong solutions. by a standard integration by parts.

It is clear that if $v \in L^2([0, T];\H^2(U)) \cap H^1([0, T]; \H^1(U)) \cap \H^2([0,T]; \L^2(U)) = \H^2(U_T)$ satisfies the condition
\be
\begin{array}{rcl}
\left . \rho^\alpha v\right |_{\dU_T} = g_0 & &\quad \textrm{for inhomogeneous Dirichlet, or} \\ \\
\left . \rho^{1-\alpha} a_{ij}\da_j v\right |_{\dU_T} = g_1& & \quad \textrm{for inhomogeneous Neumann}.
\end{array}
\ee
then we can apply our previous weak well posedness results to the functions
\be
\tilde{u} = u - v
\ee
We'd like to know what conditions are required on $g_0, g_1$ in order that such a $v$ exists. The following Lemma gives the results we require, and comes from adapting Lemma \ref{tracelem} to $\H^2(U)$.
\begin{Lemma}
\begin{enumerate}[(i)]
\item Suppose $v \in \H^2(U_T)$, then $\left . \rho^\alpha v\right |_{\dU_T}$ and $\left . \rho^{1-\alpha} a_{ij}\da_j v\right |_{\dU_T}$ exist in a trace sense, and we have
\be
\norm{\left . \rho^\alpha v\right |_{\dU_T}}{H^{1+\alpha}(\dU_T)} + \norm{\left . \rho^{1-\alpha} a_{ij}\da_j v\right |_{\dU_T}}{H^{1-\alpha}(\dU_T)} \leq C \norm{v}{\H^2(U_T)}
\ee 
\item Suppose $v_0 \in H^{1+\alpha}(\dU_T)$ and $v_1 \in H^{1-\alpha}(\dU_T)$. Then there exists $\tilde v\in \H^2(U_T)$ such that
\be
\left . \rho^\alpha \tilde v\right |_{\dU_T} = v_0 \qquad \left . \rho^{1-\alpha} a_{ij}\da_j \tilde v\right |_{\dU_T} = v_1
\ee
where the restriction is understood in the trace sense. Furthermore we may choose $\tilde v$ such that 
\be
\norm{\tilde v}{\H^2(U_T)} \leq C\left( \norm{v_0}{H^{1+\alpha(\dU_T)}} + \norm{v_1}{H^{1-\alpha}(\dU_T)}\right)
\ee 
with $C$ independent  of $v_i$.
\end{enumerate}
\end{Lemma}
Taking this with our previous results, we conclude
\begin{Theorem} \label{ExistInhom}
\begin{enumerate}[(i)]
\item Given $u_0, u_1, \f$ as above, $g_0 \in H^{1+\alpha}(\dU_T)$ with $\left. \rho^\alpha u_0 \right|_{\dU_T}=\left .g_0\right|_{t=0}.$, there exists a unique weak solution to the Dirichlet IBVP corresponding to this data, with the estimate
\be
\norm{\u}{H^1_{\textrm{sol.}, D}(U)} \leq C \left (\norm{(u_0, u_1, \f)}{H^1_{\textrm{data}}(U)} + \norm{g_0}{H^{1+\alpha}(\dU_T)}\right)
\ee
\item Given $u_0, u_1, \f$ as above, $g_1 \in H^{1-\alpha}(\dU_T)$, there exists a unique weak solution to the Neumann IBVP corresponding to this data with the estimate
\be
\norm{\u}{H^1_{\textrm{sol.}, N}(U)} \leq C \left (\norm{(u_0, u_1, \f)}{H^1_{\textrm{data}}(U)} + \norm{g_1}{H^{1-\alpha}(\dU_T)}\right )
\ee
\end{enumerate}
\end{Theorem}
We shall not go through the proof in detail, but it is clear that the higher regularity results of \S \ref{HR} can be extended to the inhomogeneous case.

\subsection{Robin boundary condition}

Another possibility for a well posed boundary condition is what we might call a Robin boundary condition:
\ben{Robin}
\rho^{1-\alpha} n_i a_{ij} \da_j u +  \beta \rho^\alpha u = 0 \textrm{ on } \dU
\een
where $\beta$ is some suitable function on $\dU_T$. We assume $\beta\in C^\infty(\dU_T)$ to be concrete, but this is not necessary. This can be achieved in the weak formulation in a similar fashion to the introduction of an inhomogeneity for the Neumann condition.
\begin{Definition}[Weak Inhomogeneous Robin IBVP]
Suppose $u_0, u_1, \f$ are as in \S \ref{inhom}. We say that $\u \in L^\infty([0, T]; \H^1(U))$ with $\dot \u \in L^\infty([0, T]; \L^2(U))$, $ \ddot \u \in L^2([0, T]; (\H^1(U))^*)$ is a \emph{weak solution} of the Robin IBVP:
\be
\begin{array}{rcl}
u_{tt} + \mathcal{L} u = f & \quad & \textrm{ in } U_T \\
\left .\left ( \rho^{1-\alpha}n_i a_{ij} \da_j u +\beta \rho^\alpha u \right )\right |_{\dU_T} = 0 & \quad & \textrm{ on } \dU_T \\
u=u_0, {u_t} = u_1 & \quad & \textrm{ on } \{0 \} \times U
\end{array}
\ee
provided
\begin{enumerate}[i)]
\item For all $v \in \H^1(U)$ and a.e. time $0 \leq t \leq T$ we have
\be
\pair{\ddot{\u}}{v} + B[\u,v;t] + \ip{\left. \rho^\alpha\mathbf{ u }\right|_{\dU}}{\left. \rho^\alpha v\right|_{\dU}}{L^2(\dU)}= \ip{\mathbf{f}}{v}{\L^2(U)}.
\ee
\item We have the initial conditions
\be
\u(0) = u_0, \qquad \dot{\u}(0) = u_1.
\ee
\end{enumerate}
\end{Definition}

It is straightforward to show that the well posedness and regularity results of \S \ref{WP} and \S \ref{HR} can be extended, where we require the estimate \eq{trest} to deal with the surface terms which arise. 

\textbf{Remark:} In the case where $b_i=0$, $c\geq 0$ and $a_{ij}, c$ independent of time, we can relate our result to the theory of essentially self-adjoint operators. A consequence of Stone's Theorem (see for example \cite{Reed}) for self-adjoint operators states:
\begin{Theorem}
Let $L: D \to H$ be a densely defined positive symmetric operator. Suppose for every $f, g\in D$ there exists a twice continuously differentiable solution (in $D$) to
\be
u_{tt} + Lu = 0; \qquad u(0)=f; \qquad u_t(0)=g
\ee  
Then $L$ is essentially self-adjoint.
\end{Theorem}
In our case $H= \L^2(U)$, $L=\mathcal{L}$, however this operator is not essentially self-adjoint on $C^\infty_0(U)$. Thus, the choice of dense subspace $D$ determines a self-adjoint extension of $\mathcal{L}$. If we take
\be
D = \{u = \rho^{\alpha} v : v \in C^\infty(\Ub), \rho^\alpha \mathcal{L}^{k}u =0 \textrm{ on } \dU \textrm{ for } k=1, 2, \ldots \},
\ee
this gives the self-adjoint extension corresponding to Dirichlet boundary conditions, whereas if we take
\be
D = \{u = \rho^{-\alpha} v : v \in C^\infty(\Ub), \rho^{1-\alpha} n_{ij}\da_j \mathcal{L}^{k}u =0 \textrm{ on } \dU \textrm{ for } k=1, 2, \ldots \},
\ee
we have the self-adjoint extension corresponding to Neumann boundary conditions. Finally, taking
\be
D = \{u = \rho^{-\alpha} v_- + \rho^\alpha v_+ : v_\pm \in C^\infty(\Ub), (\rho^{1-\alpha} n_{ij}\da_j +\rho^\alpha \beta )\mathcal{L}^{k}u =0 \textrm{ on } \dU \textrm{ for } k=0, 1, \ldots \},
\ee
gives the self-adjoint extension corresponding to the Robin boundary conditions. It is straightforward to check that $\mathcal{L}$ is positive and symmetric in all three cases, provided for the Robin case we take $\beta \geq 0 $.

As a consequence, we can apply the functional analytic machinery of essentially self-adjoint operators to $\mathcal{L}$ in these situations. In fact, we can do better than this, based on the close analogy with the finite case. It can be shown that $\mathcal{L}$ with homogeneous Dirichlet, Neumann or Robin boundary conditions has a countable set of eigenvalues with corresponding eigenfunctions, smooth in the interior of $U$, which form an orthonormal basis for $\L^2(U)$. This result comes from first establishing that $\H^1(U)$ is compactly embedded in $\L^2(U)$, and using this fact to apply the Fredholm alternative to a suitably chosen compact operator.

\end{document}